\documentclass[a4paper
               ,10pt
               ,DIV=10
               ,oneside
               ,aps
               ,notitlepage
               ,nofootinbib
               ,superscriptaddress
               ,tightenlines
               ]
               {scrartcl}

\usepackage{amsmath}
\usepackage{amssymb}  
\usepackage{amsthm}
\usepackage{amsfonts}
\usepackage{graphicx}
\usepackage{bbm}
\usepackage{url}
\usepackage{braket}
\usepackage{ upgreek }
\usepackage{dsfont}
\usepackage{authblk}

\usepackage{hyperref}
\hypersetup{
    colorlinks,
    citecolor=black,
    filecolor=black,
    linkcolor=black,
    urlcolor=black
}

\theoremstyle{plain}               
\newtheorem{thm}{Theorem}[section]
\newtheorem{lem}{Lemma}[section]
\newtheorem{cor}{Corollary}[section]

\newtheorem{defn}{Definition}[section]

\theoremstyle{remark}
\newtheorem{rem}{Remark}[section]

\def\beq{\begin{equation}}
\def\eeq{\end{equation}}
\def\bq{\begin{quote}}
\def\eq{\end{quote}}
\def\ben{\begin{enumerate}}
\def\een{\end{enumerate}}
\def\bit{\begin{itemize}}
\def\eit{\end{itemize}}

\def\ra{\rightarrow}

\def\lb{\left(}
\def\rb{\right)}
\def\lset{\lbrace}
\def\rset{\rbrace}
\def\lk{\left\langle}
\def\rk{\right\rangle}
\def\l|{\left|}
\def\r|{\right|}
\def\lbr{\left[}
\def\rbr{\right]}
\newcommand\C{\mathbbm{C}}

\newcommand\Z{\mathbbm{Z}}
\newcommand\R{\mathbbm{R}}
\newcommand\N{\mathbbm{N}}
\newcommand\M{\mathcal{M}}
\newcommand\D{\mathcal{D}}

\newcommand{\U}{\mathcal{U}}

\newcommand{\Tm}{\mathcal{T}}

\newcommand{\Em}{\mathcal{E}}

\newcommand{\Lm}{\mathcal{L}}

\newcommand{\proj}[2]{| #1 \rangle\!\langle #2 |}

\newcommand{\chanSet}{\mathfrak{C}}

\newcommand{\ketbra}[1]{|#1\rangle\!\langle#1|}
\newcommand{\tr}{\text{tr}}
\newcommand{\one}{\mathds{1}}
\newcommand{\id}{\text{id}}
\newcommand{\liou}{\mathcal{L}}
\newcommand{\scalar}[2]{\langle#1|#2\rangle}

\newcommand{\Ent}{\text{Ent}}

\newcommand{\lioud}{\liou_{\text{dep}}}

\begin{document} 

\title{{Entropy Production of Doubly Stochastic Quantum Channels}}

\author{\vspace{-0.15cm}Alexander M\"uller-Hermes\thanks{muellerh@ma.tum.de}~}
\author{Daniel Stilck Fran\c{c}a\thanks{dsfranca@mytum.de}~}
\author{Michael M.~Wolf\thanks{wolf@ma.tum.de}~}
\affil{\vspace{-0.15cm}\small{Department of Mathematics, Technische Universit\"at M\"unchen, 85748 Garching, Germany}}

%
%

\date{\today}
\maketitle


\begin{abstract}

We study the entropy increase of quantum systems evolving under primitive, doubly stochastic Markovian noise and thus converging to the maximally mixed state. 
This entropy increase can be quantified by a logarithmic-Sobolev constant of the Liouvillian generating the noise. We prove a universal lower bound on this constant that stays invariant under taking tensor-powers. 
Our methods involve a new comparison method to relate logarithmic-Sobolev constants of different Liouvillians and a technique to compute logarithmic-Sobolev inequalities of Liouvillians with eigenvectors forming a projective representation of a finite abelian group. 
Our bounds improve upon similar results established before and as an application 
we prove an upper bound on continuous-time quantum capacities. In the last part of this work we study entropy production estimates of discrete-time doubly-stochastic quantum channels by extending the framework of discrete-time logarithmic-Sobolev inequalities to the quantum case.    
\end{abstract}

\tableofcontents

\section{Introduction}

Consider a quantum system affected by Markovian noise driving every initial state to the maximally mixed state. For such a, so-called doubly stochastic and primitive, noise the von-Neumann entropy $S$ will steadily increase in time. Here we want to quantify how much entropy is produced by such a noise channel. More specifically let the Markovian noise channel be modeled by a quantum dynamical semigroup $T_t = e^{t\Lm}$ and recall that the von-Neumann entropy of a state $\rho$ is given by $S(\rho)=-\text{tr}(\rho\log(\rho))$. We want to establish bounds of the form  
\begin{align}
S(T_t(\rho))-S(\rho)\geq C_t(log(d)-S(\rho))
\label{equ:GenEntProd}
\end{align}
for some time-dependent $C_t < 1$ independent of the state $\rho$. To find good $C_t$ for a given noise channel $T_t$ in the above bound we apply the framework of logarithmic Sobolev (LS) inequalities~\cite{Olkiewicz1999246,Kastoryanosob}. For special channels bounds of the form of \eqref{equ:GenEntProd} have been considered in ~\cite{benor}. Similar bounds have also been studied in terms of contractive properties of the channel with respect to different norms in \cite[cf. Remark \ref{compotherbounds}]{streater,raginsky2002entropy}. However, in these cases the lower bound is given in terms of the trace or Hilbert-Schmidt distance between the state $\rho$ and the maximally mixed state, which are upper bounded by a constant independent
of the dimension, while the left-hand side of \eqref{equ:GenEntProd} is of order $\log\lb d\rb$. The entropy production of quantum dynamical semigroups has also been investigated in a more general setting in \cite{spohn}.

For many applications in quantum information theory it will be important to quantify the entropy production of tensor-powers $T^{\otimes n}_t$ of the noise channel. In our main result we obtain bounds of this form, where $C_t$ does not depend on $n$. We improve on recent results by Temme et al.~\cite{2014arXiv1403.5224T} who established such bounds using spectral theory. The invariance under taking tensor-powers makes these bounds important for many applications including the study of stability in Markovian systems~\cite{stabcubitt}, mixing-time bounds~\cite{2014arXiv1403.5224T} and quantifying the storage time in quantum memories~\cite{subdivisioncapacities}.

In the second part of this paper we consider entropy production estimates of the form \eqref{equ:GenEntProd} for discrete-time doubly stochastic quantum channels. We introduce the framework of discrete LS inequalities, which allows us to generalize results from classical Markov  chains, where there already is a vast literature on the subject~\cite{diaconis1996,diaconiscomparison,bobkov2006modified,miclo}.

This paper is organized as follows:
\begin{itemize}
 \item In section \ref{sec:2} we introduce our notation, definitions and show how the framework of logarithmic Sobolev inequalities relates to the entropy production of a doubly stochastic Markovian time-evolution. 
\item In section \ref{sec:3} we prove our main result. This is an improved lower bound on the LS constant of tensor powers of doubly stochastic semigroups (Theorem \ref{lowerboundtensor}), which directly implies an entropy production estimate (Corollary \ref{cor:EntrProdTens}) for tensor-products of Markovian time-evolutions. Previous approaches to this problem focused on spectral and interpolation techniques~\cite{hyperspec,2014arXiv1403.5224T}. Here we obtain better bounds with simpler proofs using 
group theoretic techniques similar to the ones developed in \cite{junge} and comparison inequalities. 
\item In section \ref{sec:5} we consider Liouvillians of the form $\liou=T-\id$, where $T$ is a quantum channel. We show how to use LS constants of classical Markov chains to analyze the entropy production of such semigroups. As an application of our techniques we compute the LS constant of all doubly stochastic qubit Liouvillians of this form.
\item In section \ref{sec:4} we extend techniques from LS inequalities to analyze discrete-time quantum channels. Here we not only get bounds on the entropy production (Theorem \ref{timproveddata}) but also on the hypercontractivity (Theorem \ref{thm:DiscreteHyper}) of these channels. However, the obtained bounds are in general weaker and become trivial as we increase the number of copies of the channel. These results are mostly a generalization of \cite{miclo}.
\item In section \ref{sec:6} we apply the results from section \ref{sec:2} to unitary quantum subdivision capacities introduced by some of the authors~\cite{subdivisioncapacities}. We show (Theorem \ref{thm:UnitarySubDCapDepolBound}) that the unitary quantum subdivision capacity 
of any doubly stochastic and primitive Liouvillian has to decay exponentially in time. Our bound improves similar results found in~\cite{benor,ben2013quantum}. In the second part of the section we compute entropy production estimates 
for random Pauli channels.
\end{itemize}

\section{Notations and preliminaries}
\label{sec:2}

Throughout this paper $\M_d$ will denote the set of complex $d\times d$-matrices and $\M_d^+\subset \M_d$ the cone of strictly positive matrices. The set of $d$-dimensional density matrices or states, i.e. positive matrices in $\M_d$ with trace $1$, will be denoted by $\D_d$ and the set of strictly positive states will be denoted by $\D_d^+$. The $d\times d$ identity matrix will be denoted
by $\one_d$. 

We will call a completely positive, trace preserving linear map $T:\M_d\to\M_d$
a quantum channel and will denote its adjoint with respect to the Hilbert-Schmidt scalar product by $T^*$. A quantum channel $T$ is said to be doubly stochastic if
 \begin{equation*}
 T\lb \one_d\rb=T^*\lb\one_d\rb=\one_d. 
 \end{equation*}

A family of quantum channels $\{T_t\}_{t\in\R_+}$ parametrized by a nonnegative parameter will be called a 
quantum dynamical semigroup if $T_0=\id_d$ (the identity map in $d$ dimensions), $T_{s+t}=T_sT_t$ for any $s,t\in\R_+$ and $T_t$ depends continuously on $t$. 
Physically a quantum dynamical semigroup describes a Markovian evolution in continuous time. It is well known~\cite{lindblad1976,gorini} 
that any quantum dynamical semigroup is generated by a Liouvillian $\liou:\M_d\ra\M_d$ of the form 
\begin{align*}
\liou(X)=\Phi(X)-\kappa X-X\kappa^\dagger,
\end{align*}
for $\kappa\in\M_d$ and $\Phi:\M_d\to\M_d$ completely positive such that $\Phi^*(\one_d) = \kappa + \kappa^\dagger$. 

A quantum dynamical semigroup generated by a Liouvillian $\liou$ 
is said to have a fixed point $\sigma\in\D_d$ if $\liou(\sigma)=0$.
Then all quantum channel in $\{e^{t\liou}\}_{t\in\R_+}$ have $\sigma$ as a fixed point. In the special case where the fixed point is $\sigma = \frac{\one_d}{d}$ we call the quantum dynamical semigroup and its Liouvillian doubly stochastic. 
If the generator is hermitian with respect to the Hilbert-Schmidt scalar product, i.e. $\liou=\liou^*$, we will call it reversible. We will be interested in the asymptotic behavior of quantum dynamical semigroups. A quantum dynamical semigroup $T_t:\M_d\ra\M_d$ with a full-rank fixed point $\sigma\in\D_d$ is called primitive if $\lim\limits_{t\to\infty}T_t(\rho)=\sigma$
for all states $\rho\in\D_d$. In the following we will be interested in particular in tensor-products of quantum dynamical semigroups. Given a Liouvillian $\Lm:\M_d\ra\M_d$ we denote by 
\begin{align}
\Lm^{(n)} := \sum^{n}_{i=1} \id_d^{\otimes i-1}\otimes\liou\otimes \id_d^{\otimes (n-i)}
\label{equ:tensorLiouv}
\end{align}
the generator of the quantum dynamical semigroup $(e^{t\Lm})^{\otimes n}$.

We will need distance measures on the set $\M_d$. Recall the family of Schatten $p$-norms for $p\in\lbr 1,\infty\rb$ defined as 
\begin{align*}
\| X\|_p := \lb\sum^d_{i=1} s_i(X)^p\rb^{1/p}
\end{align*}
where $s_i(X)$ denotes the $i$-th singular value of $X\in\M_d$ and $s(X)\in\R^d$ is the 
vector containing the ordered singular values of $X$ as entries. Note that we can consistently define $\| X\|_\infty := \sup_{i\in\lset 1,\ldots ,d\rset} s_i(X)$ for any $X\in\M_d$.

Another distance measure (although not a metric) on the set $\D_d$ of states is the relative entropy (also known as Kullback-Leibler divergence):
\begin{align*}
 D(\rho\|\sigma)=
\begin{cases} 
\tr[\rho(\log\rho-\log\sigma)], & \mbox{if }  \mbox{supp}(\rho)\subset\mbox{supp}(\sigma)
\\ +\infty, & \mbox{otherwise}
\end{cases}
\end{align*}
for $\rho,\sigma\in\D_d$. Recall Pinsker's inequality: 
\begin{align*}
 D(\rho\|\sigma)\geq\frac{1}{2}\|\rho-\sigma\|_1^2
\end{align*}
for any $\rho,\sigma\in \D_d$. This inequality implies in particular that $D(\rho\|\sigma)=0$ iff $\rho=\sigma$.

\subsection{The LS-1 Constant}
Consider an inequality of the form:
 \begin{equation}\label{expdecayrel}
  D\left(T_t(\rho)\Big{\|}\frac{\one_d}{d}\right)\leq e^{-2\alpha t}D\left(\rho\Big{\|}\frac{\one_d}{d}\right)
 \end{equation}
for some doubly stochastic Liouvillian $\liou:\M_d\to\M_d$, and where $T_t=e^{t\liou}$, $\rho\in\D_d$ and $\alpha\in\R_+$ is a constant independent of $\rho$. Using $D\left(\rho\big{\|}\frac{\one_d}{d}\right)=\log(d)-S(\rho)$ the above inequality is clearly equivalent to \eqref{equ:GenEntProd} for $C_t = (1- e^{-2\alpha t})$. The framework of logarithmic Sobolev-1-inequalities (LS-1 inequality) allows us to determine the optimal $\alpha$ such that \eqref{expdecayrel} holds. To this end define the function $f(t)=D\left(T_t(\rho)\big{\|}\frac{\one_d}{d}\right)$. If we can show 
\begin{equation}\label{difineq1}
\frac{df}{dt}\leq-2\alpha f
\end{equation}
for some $\alpha\in\R_+$ it follows that $f(t)\leq e^{-2\alpha t}f(0)$. The time derivative of the relative entropy at $t=0$, also called the entropy production\cite{spohn},
is given by:
\begin{align*}
 \frac{d}{dt} D\left(T_t(\rho)\Big{\|}\frac{\one_d}{d}\right)\bigg{|}_{t=0} &=\tr[\liou(\rho)(\log(d)+\log(\rho))] \\
 &=\tr[\liou(\rho)\log(\rho)]
\end{align*}
as $\tr(\liou(\rho)) = 0$ for any $\rho\in\D_d$.
This motivates the definition of the LS-1 constant:
\begin{defn}[LS-1 constant]
 Let $\liou:\M_d\to\M_d$ be a doubly stochastic Liouvillian. We define its \textbf{LS-1 constant} as
 \begin{align}
  \alpha_1(\liou) :=\inf\limits_{\rho\in\D_d^+}-\frac{1}{2}\frac{\tr[\liou(\rho)\log(\rho)]}{D(\rho\big{\|}\frac{\one_d}{d})}
  \label{alpha1Def}
 \end{align}
\label{defn:LS-1}
\end{defn}
By the above discussion \eqref{expdecayrel} is valid for $\alpha = \alpha_1(\liou)$ as it is true for small times and the time can be extended by iterating the bound. Also by definition it is the optimal constant such that \eqref{expdecayrel} holds independent of $\rho$. 
Note that we may only consider states with full rank in the optimization of Definition \ref{defn:LS-1} due to the continuity of the relative entropy and entropy production. Also note that if the Liouvillian $\Lm$ has another fixed point different from $\frac{\one_d}{d}$, then $\alpha_1(\liou) = 0$ and \eqref{expdecayrel} reduces to the data processing inequality. In the following we will always consider primitive Liouvillians and thereby avoid this issue.

Using $D\left(\rho\big{\|}\frac{\one_d}{d}\right)=\log(d)-S(\rho)$ the following theorem follows from \eqref{expdecayrel}:

\begin{cor}[Entropy increase]\label{entdiff}
 Let $\liou:\mathcal{M}_d\to\mathcal{M}_d$ be a doubly stochastic Liouvillian and $\alpha\leq \alpha_1\lb\liou\rb$ then
 \begin{align*}
  S(T_t(\rho))-S(\rho)\geq(1-e^{-2\alpha t})(log(d)-S(\rho))
 \end{align*}

 for any $\rho\in\D_d$ and where $T_t=e^{\liou t}$ denotes the semigroup generated by $\liou$.
\end{cor}

\subsection{The LS-2 Constant}

The non-linearity of the optimization problem in \eqref{alpha1Def} makes it hard to compute the LS-1 constant analytically. In fact there are only few examples even for classical Markov chains~\cite{bobkov2006modified} where the LS-1 constant is known. For many applications, however, it will be enough to have good lower bounds on $\alpha_1$. Here these lower bounds will be in terms of the so-called LS-2 constant:
 
 \begin{defn}[LS-2 constant]
Let $\liou:\M_d\to\M_d$ be a doubly stochastic Liouvillian. We define its \textbf{LS-2 constant} as
\begin{align*}
 \alpha_2(\liou):=\inf\limits_{X\in\M_d^+}\frac{\Em^2_\liou(X)}{\Ent_2(X)}.
\end{align*}
Here we used the so-called \textbf{2-Dirichlet form}  
\begin{align*}
\Em^2_\liou(X):=-\frac{1}{d}\tr[\liou(X)X]
\end{align*}
 and the so-called \textbf{2-relative entropy} given by
 \begin{equation*}
  \Ent_2(X):=\frac{1}{2d}\tr\left[ X^2\lb\log\lb\frac{X^2}{\tr(X^2)}\rb+\log(d)\rb\right]
 \end{equation*}
\label{defn:LS2}
\end{defn}

It is well known that a Liouvillian is primitive iff we have a unique strictly positive density matrix in the kernel of $\liou$. From this it is easy to see that $\alpha_2(\liou)=0$ 
if the Liouvillian is not primitive, as in this case there exists $X\in\M_d^+$ s.t. $X\not\in\text{span}\{\one\}$ and $\Em_2^\liou(X)=0$. We will later see that the LS-2 constant is strictly positive if the Liouvillian is primitive.

As the 2-Dirichlet form is bilinear, $\alpha_2$ is easier to compute than $\alpha_1$, where a logarithm occurs in the numerator (see \eqref{alpha1Def}). Also by $\alpha_2(\liou)=\alpha_2(\frac{\liou+\liou^*}{2})$ we may always suppose that the Liouvillian is reversible when computing $\alpha_2$. Another advantage of the LS-2 constant in comparison to the LS-1 constant is the following hypercontractive characterization. This characterization allows the use of tools from other areas of mathematics, such as interpolation theory, to compute the LS-2 constant. 

\begin{thm}[Hypercontractive picture~\cite{Olkiewicz1999246}]\label{hypercontrpic}
 Let $\liou:\M_d\to\M_d$ be a primitive doubly stochastic Liouvillian and $T_t=e^{t\liou}$ its associated semigroup. Then:
 \begin{enumerate}
 \item If there exists $\alpha>0$ such that $\|T_t(X)\|_{p(t),\frac{\one_d}{d}}\leq\|X\|_{2,\frac{\one_d}{d}}$ for all $X\in\M_d^+$, where 
  $p(t)=1+e^{2\alpha t}$, it follows that $\alpha_2(\liou)\geq\alpha$.
  \item For $\alpha_2(\liou)>0$, 
  we have $\|T_t(X)\|_{p(t),\frac{\one_d}{d}}\leq\|X\|_{2,\frac{\one_d}{d}}$ for all $X\in\M_d^+$,
  with $p(t)=1+e^{2\alpha_2 t}$ for $\liou$ reversible and $p(t)=1+e^{\alpha_2 t}$ if not. 
 \end{enumerate}
 
 Here we used the $\frac{\one_d}{d}$-\textbf{weighted} $l_p$-\textbf{norm} on $\M_d$ given by:
\begin{align*}
\| Y\|_{p,\frac{\one_d}{d}}:=\frac{1}{d^{\frac{1}{p}}}(\tr[|Y|^p])^{\frac{1}{p}}  
\end{align*}
  
\end{thm}  

We will state the connection between the LS-2 and the LS-1 constants in Theorem \ref{relationconstants} below.

\subsection{The spectral gap and relations between the LS-constants }
Another important constant for studying the convergence properties of quantum dynamical semigroups is the spectral gap. 
Usually a unital Liouvillian $\liou:\M_d\to\M_d$  is said to have a spectral gap $\lambda'$ if the $0$ eigenvalue corresponding to $\frac{\one_d}{d}$ is the only eigenvalue with real part $0$ 
whereas $|\text{Re}\lambda_i|\geq\lambda'$ for all other eigenvalues of $\liou$.
In the context of LS-inequalities the following definition is used:

\begin{defn}[Spectral Gap]
 Let $\liou:\M_d\to\M_d$ be a doubly stochastic Liouvillian. The \textbf{spectral gap} is defined as:
 \begin{align*}
  \lambda(\liou):=\inf\limits_{X\in\M_d:X=X^\dagger}\frac{\Em_\liou^2(X)}{\text{Var}(X)}.
 \end{align*}
 Here we used the \textbf{variance} with respect to the maximally mixed state defined as
 \begin{align*}
 \text{Var}_{\frac{\one_d}{d}}(Y):=\|Y-\tr(Y)\frac{\one_d}{d}\|_{2,\frac{\one_d}{d}}^2 ~.
 \end{align*}

\end{defn}
It agrees with the usual definition for reversible Liouvillians and is the spectral gap of the additive 
symmetrization $\frac{\liou+\liou^*}{2}$. Indeed, for reversible Liouvillians
the spectrum is nonpositive and we may assume w.l.o.g. that the eigenvector $X_i$ corresponding to an eigenvalue $\lambda_i$
is hermitian. By the orthogonality of eigenvectors, we have that $\tr[\one_d X_i]=0$, thus all eigenvectors that do not correspond to the eigenvalue $0$ are also
traceless and are invariant under the transformation $X\mapsto\ X-tr[X]\frac{\one}{d}$. By these considerations,
\begin{align}
\inf\limits_{X\in\M_d:X=X^\dagger}\frac{\Em_\liou^2(X)}{\text{Var}(X)} 
\end{align}
is the second largest eigenvalue of $-\liou$ if $\liou$ is reversible or the of $-\frac{\liou+\liou^*}{2}$ if not, coinciding with the usual definition.

%
Finally we can state how the LS-constants and the spectral gap relate to each other.

\begin{thm}[\cite{Kastoryanosob}]\label{relationconstants}
 Let $\liou:\M_d\to\M_d$ be a doubly stochastic  Liouvillian. If $\liou$ is reversible, the spectral gap, LS-2 and LS-1 constant satisfy:
 \begin{align*}
  \lambda(\liou)\frac{2(1-\frac{2}{d})}{\log(d-1)}\leq\alpha_2(\liou)\leq\alpha_1(\liou)\leq\lambda(\liou)
 \end{align*}
If $\liou$ is not reversible they satisfy:
\begin{align*}
  \lambda(\liou)\frac{(1-\frac{2}{d})}{\log(d-1)}\leq\frac{\alpha_2(\liou)}{2}\leq\alpha_1(\liou)\leq\lambda(\liou)
\end{align*}
For $d=2$ the function $d\mapsto\frac{2(1-\frac{2}{d})}{\log(d-1)}$ can be extended continuously by $1$.
\end{thm}

The previous theorem establishes a connection between the LS-1 and LS-2 constants. We will use this connection as it is usually easier to derive bounds on the LS-2 constant than on the LS-1 constant directly. Note that by combining Theorem \ref{relationconstants} with Theorem \ref{entdiff} we immediately obtain the bound
\begin{align}
S(T_t(\rho))-S(\rho)\geq(1-e^{-\alpha_2\lb \liou\rb t})(\log(d)-S(\rho))
\label{equ:entrProdAlpha2}
\end{align}
for any $\rho\in\D_d$ and where $T_t=e^{\liou t}$ denotes the semigroup generated by $\liou$.

%
%

\section{Continuous LS inequalities for doubly stochastic Liouvillians}
\subsection{Tensor-stable LS-inequalities}
\label{sec:3}

For doubly stochastic and primitive Liouvillians $\Lm:\M_d\ra\M_d$ we consider the generator $\Lm^{(n)}$ (see \eqref{equ:tensorLiouv}) of the tensor-product semigroup $\lb e^{t\Lm}\rb^{\otimes n}$. 
We will prove a lower bound on $\alpha_2\lb\Lm^{(n)}\rb$ that does not depend on $n$. By \eqref{equ:entrProdAlpha2} such bounds directly lead to entropy production inequalities for tensor-products of quantum dynamical semigroups 
(see Corollary \ref{cor:EntrProdTens}). These inequalities turn out to be useful for the analysis of quantum memories (see section \ref{sec:subd}). 

For our bounds on the LS-2 constant we will first compute a lower bound on $\alpha_2\lb\lioud^{(n)}\rb$, where $\Lm_{\text{dep}}:\M_d\ra\M_d$ denotes the depolarizing Liouvillian given by
\begin{align}
\lioud(X)=\tr(X)\frac{\one_d}{d}-X ~.
\label{equ:depLiou}
\end{align}
Then we use a comparison technique to derive the desired bound for general doubly stochastic and primitive Liouvillians.

We will need the following theorem proved in \cite{2014arXiv1403.5224T} showing that it suffices to show hypercontractivity for one fixed time to
lower bound $\alpha_2(\liou)$:
\begin{thm}\cite[Theorem 5]{2014arXiv1403.5224T}\label{snapshothyper}
 Let $\liou:\M_d\to\M_d$ be a reversible, doubly stochastic Liouvillian with spectral gap $\lambda$. Suppose that for some
 $t_0\in\R_+$ we have $\|T_{t_0}\|_{2\to 4,\frac{\one_d}{d}}\leq1$. Then:
 \begin{align*}
  \alpha_2(\liou)\geq\frac{\lambda}{4\lambda t_0+2}
 \end{align*}

\end{thm}

Using the group theoretic techniques and definitions introduced in the appendix we prove the following bound on the $2\to 4$ norm of tensor powers of the depolarizing channel. Similar bounds have been developed in \cite{junge}.

\begin{thm}\label{upper2to4dep}
Let $T_t:\M_d\ra\M_d$ denote the semigroup $T_t = e^{t\lioud}$ generated by $\lioud:\M_d\to\M_d$ as defined in \eqref{equ:depLiou}. Then we have
\begin{align}
\|T_{t_0}^{\otimes n}\|_{2\to4,\frac{\one}{d^n}}\leq1
\end{align}
for $t_0=\frac{\log(3)\log(d^2-1)}{4\lb1-2d^{-2}\rb}$.
\end{thm}
\begin{proof}
Note that the Weyl system \eqref{equ:Weyl} forms an almost commuting unitary eigenbasis (see Definition \ref{defn:AlmostCommut}) for the depolarizing Liouvillian $\lioud$. The unitaries of the Weyl system can be associated to characters on $\Z_d\times \Z_d$. As explained in the appendix we can associate a classical semigroup $P_t$ (see \eqref{equ:Pt}) acting on the space $V(\Z_d\times \Z_d)$ of complex functions on $\Z_d\times \Z_d$. It is easy to verify that the generator $L$ of this classical semigroup coincides with the generator of the random walk on the complete graph with $d^2$ vertices and uniform distribution.
In \cite[Theorem A.1]{diaconis1996} it was shown that:
\begin{align}
\alpha_2(L)=\frac{2\lb1-2d^{-2}\rb}{\log(d^2-1)} 
\end{align}
Also it is known that for classical semigroups~\cite[Lemma 3.2]{diaconis1996}
\begin{align}
\alpha_2(L_1\otimes\id+\id\otimes L_2)=\min\{\alpha_2(L_1),\alpha_2(L_2)\} .
\end{align}
Thus, by the hypercontractive characterization of the LS-2 constant\cite[Theorem 3.5]{diaconis1996} we have
\begin{align*}
\|P_t^{\otimes n}\|_{2\to p(t)}\leq1, 
\end{align*}
for any $n\in\N$ where $p(t)=1+e^{2\alpha_2(L)t}$. With $t_0=\frac{\log(3)}{2\alpha_2(L)}$ we have $p(t_0)=4$ and the
the claim follows if we apply Theorem \ref{uppertensors} inductively.
\end{proof}

As the spectral gap of $\lioud$ is $1$ we obtain the following corollary by applying Theorem \ref{upper2to4dep} and Theorem \ref{snapshothyper}.

\begin{cor}[Lower bound on LS-2 for tensor powers of the depolarizing channel]\label{lowerboundtensordep}
 Let $\lioud:\M_d\to\M_d$ be the depolarizing Liouvillian \eqref{equ:depLiou}. Then:
 \begin{align*}
  \alpha_2(\lioud^{(n)})\geq\frac{\lb1-2d^{-2}\rb}{\log(3)\log(d^2-1)+2\lb1-2d^{-2}\rb} .
 \end{align*}

\end{cor}

Note that any doubly stochastic Liouvillian $\liou$ commutes with the depolarizing Liouvillian $\lioud$. We can use this simple observation to prove the following comparison theorem, which will lead to our main result.   

\begin{thm}[Comparison with the depolarizing Liouvillian]\label{lowerboundtensor}
Let $\liou:\M_d\to\M_d$ be a doubly stochastic and primitive Liouvillian with spectral gap $\lambda$ and $\| \frac{\liou+\liou^*}{2}\|$ 
the operator norm of its additive symmetrization. For any $n\in\N$ we have 
 \begin{align*}
  \Big{\|} \frac{\liou+\liou^*}{2}\Big{\|}\alpha_2(\lioud^{(n)})\geq\alpha_2(\liou^{(n)})\geq\lambda\alpha_2(\lioud^{(n)})
 \end{align*}
 where $\liou^{(n)}$ is defined as in \eqref{equ:tensorLiouv}.

\end{thm}
\begin{proof}
When working with the LS-2 constant (see Definition \ref{defn:LS2} and the discussion following this definition) we may consider the additive symmetrization $\frac{\liou+\liou*}{2}$
instead of $\liou$. Therefore, we may assume that $\liou$ is reversible without loss of generality.

Using $\tr(\liou(Y))=0$ for all $Y\in\M_d$ and $\liou(\one_d)=0$ it is easily seen that $\lbr\lioud,\liou\rbr = 0$. This shows that $\liou$ and $\lioud$ can be simultaneously diagonalized. The same also holds for $\liou^{(n)}$ and $\lioud^{(n)}$.

Let $\{Y_i\}_{0\leq i\leq d^2-1}$ be an orthonormal basis for $\M_d$ with respect to the normalized Hilbert-Schmidt scalar product $\scalar{\cdot}{\cdot}_{\frac{\one_d}{d}} = \frac{1}{d}\scalar{\cdot}{\cdot}_{\text{HS}}$ consisting of
eigenvectors of $\liou$ and $\lioud$ with $Y_0=\one_d$.
Let $\lambda_i$ denote the corresponding eigenvalues of $\liou$, i.e. such that $\liou(Y_i)=\lambda_iY_i$.
We will show that
\begin{align*}
 \frac{1}{\|\liou\|}\Em^2_{\liou^{(n)}}(X)\leq\Em^2_{\lioud^{(n)}}(X)\leq\frac{1}{\lambda}\Em^2_{\liou^{(n)}}(X)
\end{align*}
for $X\in\M_{d^n}^+$. 

Given a multi-index $\nu\in [d^2-1]^n$, where $[d^2-1]=\{0,...,d^2-1\}$, we define
\begin{align*}
Y_\nu := \bigotimes\limits_{i=1}^{n}Y_{\nu(j)}.    
\end{align*} 
Clearly $\lset Y_\nu\rset_{\nu\in [d^2-1]^n}$ forms an orthonormal basis of eigenvectors of $\liou^{(n)}$ and $\lioud^{(n)}$. For a multi-index $\nu\in [d^2-1]^n$ we define $\text{supp}\lb\nu\rb := \lset i\in \lset 1,2,\ldots, n\rset : \nu(i)\neq 0 \rset$. As $\liou(Y_0)=\lioud(Y_0)=0$ we have 
\begin{align*}
\liou^{(n)}(Y_{\nu})=\lb\sum\limits_{i\in \text{supp}\lb\nu\rb}\lambda_{\nu(i)}\rb Y_{\nu}
\end{align*}
and 
\begin{align*}
\lioud^{(n)}(Y_{\nu})=-|\text{supp}\lb\nu\rb|Y_{\nu}
\end{align*}
for any multi-index $\nu\in [d^2-1]^n$. With
\begin{align*}
\lambda_{\nu} :=\sum\limits_{i\in \text{supp}\lb\nu\rb}\lambda_{\nu(i)}
\end{align*}
we can express the 2-Dirichlet forms (see Definition \ref{defn:LS2}) as  
\begin{align}\label{dirichlettensor}
\Em_{\liou^{(n)}}^2(X)=-\sum\limits_{\nu\in [d^2-1]^n}\lambda_{\nu}\l|\scalar{X}{Y_{\nu}}_{\frac{\one_{d^n}}{d^n}}\r|^2
\end{align}
and
\begin{align}
 \Em_{\lioud^{(n)}}^2(X)=\sum\limits_{\nu\in [d^2-1]^n}|\text{supp}\lb \nu\rb| \l|\scalar{X}{Y_{\nu}}_{\frac{\one_{d^n}}{d^n}}\r|^2.
\end{align}
We know that the spectral gap is given by $\lambda=\min\limits_{i\neq 0}\{-\lambda_i\}$ and $\|\liou\|=\max\limits_{i}\{-\lambda_i\}$. As a consequence we have
\begin{equation}\label{lowerboundeig}
 -\lambda_{\nu}\geq|\text{supp}\lb\nu\rb |\lambda,
\end{equation}
\begin{equation}\label{upperboundeig}
  -\lambda_{\nu}\leq|\text{supp}\lb\nu\rb |\|\liou\|.
\end{equation}

Combining \eqref{dirichlettensor} and \eqref{lowerboundeig} leads to
\begin{align*}
 \Em_{\liou^{(n)}}^2(X)\geq \lambda\Em_{\lioud^{(n)}}^2.
\end{align*}
By Definition \ref{defn:LS2} of the LS-2 constant this shows $\alpha_2(\liou^{(n)})\geq\lambda\alpha_2(\lioud^{(n)})$. 

In the same way combining \eqref{upperboundeig} and \eqref{dirichlettensor} implies $\alpha_2(\liou^{(n)})\leq\|\liou\|\alpha_2(\lioud^{(n)})$.

\end{proof}

The previous result tells us that, considering Liouvillians with fixed operator norms, 
the depolarizing channel is the most hypercontractive one, as it has the largest 
LS-2 constant in this class. 

One can also introduce LS constants for Liouvillians $\liou$ having a stationary state $\sigma\in\D_d^+$ that is not necessarily maximally mixed\cite{Kastoryanosob,Olkiewicz1999246}. When $\liou$ is primitive and reversible the same proof as for Theorem \ref{lowerboundtensor} yields
\begin{align*}
 \|\liou\|\alpha_2(\liou_{\text{dep},\sigma}^{(n)})\geq\alpha_2(\liou^{(n)})\geq\lambda\alpha_2(\liou_{\text{dep},\sigma}^{(n)})
\end{align*}
for the generalized depolarizing Liouvillian $\liou_{\text{dep},\sigma}(X):=\tr(X)\sigma-X$.

By combining Theorem \ref{lowerboundtensor} with Corollary \ref{lowerboundtensordep} we can finally establish an explicit lower bound on $\alpha_2(\liou^{(n)})$. For an explicit upper bound we can also apply Corollary \ref{lowerboundtensordep} to
\begin{align*}
\alpha_2(\lioud^{(n)})\leq \alpha_2(\lioud)=\frac{2(1-\frac{2}{d})}{\log(d-1)}
\end{align*} 
where we used that the LS-2 constant can only decrease when taking tensor-powers (see Definition \ref{defn:LS2}). We also used the explicit formula for the LS-2 constant of a single depolarizing Liouvillian $\lioud:\M_d\ra\M_d$ calculated in \cite{Kastoryanosob}. Summarizing these observations we obtain the following corollary:    

\begin{cor}
Let $\liou:\M_d\to\M_d$ be a doubly stochastic and primitive Liouvillian with spectral gap
 $\lambda$ and $\| \frac{\liou+\liou^*}{2}\|$ 
the operator norm of its additive symmetrization. Then:
 \begin{align*}
 \Big{\|} \frac{\liou+\liou^*}{2}\Big{\|}\frac{2(1-\frac{2}{d})}{\log(d-1)}\geq\alpha_2(\liou^{(n)})\geq\frac{\lambda\lb1-2d^{-2}\rb}{\log(3)\log(d^2-1)+2\lb1-2d^{-2}\rb}.
 \end{align*}
 \label{cor:TSBoundAlpha} 
\end{cor}

Note that the lower bound in Corollary \ref{cor:TSBoundAlpha} is slightly better than the one that follows from the results in \cite{hyperspec,2014arXiv1403.5224T} given by
\begin{align*}
\alpha_2(\liou^{(n)})\geq \frac{\lambda}{5\log(d)+11} ~.
\end{align*}

Using \eqref{equ:entrProdAlpha2} we obtain the following corollary on the entropy production of a tensor-product semigroup:

\begin{cor}[Entropy production of tensor-product semigroups]
Let $\liou:\M_d\to\M_d$ be a doubly stochastic and primitive Liouvillian with spectral gap
 $\lambda$. Then we have 
\begin{align*}
S(T^{\otimes n}_t\lb\rho\rb) - S(\rho) \geq (1-e^{-2t\alpha\lb d\rb})\lb\log(d^n) - S\lb\rho\rb\rb 
\end{align*}
with $\alpha\lb d\rb = \frac{\lambda\lb1-2d^{-2}\rb}{\log(3)\log(d^2-1)+2\lb1-2d^{-2}\rb}$ for the quantum dynamical semigroup $T_t = e^{t\Lm}$ generated by $\liou$. 

\label{cor:EntrProdTens}
\end{cor}

In \cite{king}[Theorem 3] the multiplicativity of the $2\to q$ norm for $q\geq2$ has been proven for 
doubly stochastic and reversible quantum channels $T:\M_2\to\M_2$. That is, for
a completely positive map $\Omega:\M_{d'}\to\M_{d'}$ we have
\begin{align*}
\|\Omega\otimes T\|_{2\to q}=\|\Omega\|_{2\to q}\|T\|_{2\to q} ~.
\end{align*}
It then follows from the hypercontractive characterization (Theorem \ref{hypercontrpic}) of the LS-2 that
$\alpha_2(\liou^{(n)}) = \alpha_2(\liou)$ for any doubly stochastic, primitive and reversible
qubit Liouvillian $\liou:\M_2\ra\M_2$, but for non-reversible Liouvillians it only follows that 
$\alpha_2\left(\liou^{(n)}\right)\geq\frac{\lambda}{2}$. 
 
Here we give a proof of this lower-bound without needing reversibility:

\begin{cor}
 Let $\liou:\M_2\to\M_2$ be a doubly stochastic and primitive Liouvillian with spectral gap $\lambda$. Then:
 \begin{align*}
  \alpha_2(\liou^{(n)})=\alpha_2(\liou)=\lambda
 \end{align*}
\label{cor:Qubit}
\end{cor}

\begin{proof}
 In \cite[Lemma 25]{Kastoryanosob}, it was proven that $\alpha_2(\lioud^{(n)})=\alpha_2(\lioud)=1$.
 Theorem \ref{lowerboundtensor} then implies $\alpha_2(\liou^{(n)})\geq\lambda$.
 But, by Theorem \ref{relationconstants}, $\alpha_2(\liou^{(n)})\leq\lambda$, which proves the claim.
\end{proof}

Note that if we are interested in the hypercontractivity of $\liou^{(n)}:\M_{2^n}\ra\M_{2^n}$,
this result implies
\begin{align*}
\|e^{t\liou^{\lb n\rb}}\|_{2\to p(t),\frac{\one_{2^n}}{2^n}}\leq1
\end{align*}
with $p(t)=1+e^{2\lambda t}$ if $\liou$ is reversible and with $p(t)=1+e^{\lambda t}$ if $\liou$ is not reversible.

\subsection{Qubit Liouvillians of the form $\Lm = T-\id$}
\label{sec:5}

In this section we consider quantum dynamical semigroups on a qubit system generated by Liouvillians of the form $\Lm = T-\id_2$ for a doubly stochastic quantum channel $T:\M_2\ra\M_2$. For these Liouvillians we will compute the LS-1 and LS-2 constants using the corresponding constants for classical Markov chains~\cite{diaconis1996}. By the general theory of LS-inequalities (see Section \ref{sec:2}) this leads to entropy production estimates for this kind of Liouvillians.

We start with the LS-2 constant: to establish the connection with classical Markov chains note that we can split the infimum in Definition \ref{defn:LS2} and optimize over spectrum and basis separately. Let $\U_d$ denote the set of unitary $d\times d$-matrices and denote the diagonal matrix with entries $s\in\R^d$ by $\text{diag}(s)$. For a doubly stochastic quantum channel $T:\M_d\ra\M_d$ and the Liouvillian $\liou = T -\id_d$ the definitions of $\Em^2_\liou$ and $\Ent_2$ lead to:
\begin{align}
\alpha_2(\liou)&=\inf\limits_{U\in\U_d}\inf\limits_{s\in\R_+^d}\frac{\Em^2_\liou(U\text{diag}(s)U^\dagger)}{\Ent_2(U\text{diag}(s)U^\dagger)} \\
& = \inf\limits_{U\in\U_d}\inf\limits_{s\in\R_+^d}\frac{-2\lk s\r| (M_U-\one_d)\l| s\rk}{\sum_i s^2_i (\log(\frac{s_i}{\| s\|_{2}}) + \log(d))}.
\label{equ:class}
\end{align}
Here $M_U\in\M_d$ is a doubly stochastic matrix depending on $U\in\U_d$ defined as 
\begin{align}
(M_U)_{ij} = \lk j\r|U^\dagger T(U\proj{i}{i}U^\dagger) U\l| j\rk
\label{equ:Markovkernel}
\end{align}
for some fixed orthonormal basis $\lset \l| i\rk\rset\subset \C^{d}$. Note that $M_U$ is doubly stochastic for any $U\in\U_d$ and thus the matrix $M_U - \one_d$ defines a classical Markov kernel on a $d$-point set. Finally let $\alpha^{(c)}_2(M_U-\one_d)$ denote the classical LS-2 constant~\cite[equation (3.1)]{diaconis1996} of the Markov kernel $M_U-\one_d$. By direct comparison we have
\begin{align}
\alpha_2(\liou) = \inf\limits_{U\in\U_d} \alpha^{(c)}_2(M_U-\one_d).
\label{equ:classLS2}
\end{align}

The LS-1 constant of the Liouvillian $\liou = T -\id_d$ can be treated in the same way as the LS-2 constant above. A similar reasoning then leads to
\begin{align}
\alpha_1(\liou) = \inf\limits_{U\in\U_d} \alpha^{(c)}_1(M_U-\one_d)
\label{equ:classLS1}
\end{align}
where $\alpha^{(c)}_1(M_U-\one_d)$ denotes the classical LS-1 constant~\cite[Equation 1.5]{bobkov2006modified} of the Markov kernel $M_U-\one_d$, see \eqref{equ:Markovkernel}.

The above technique works for every dimension $d\geq 2$. We will now restrict to $d=2$, i.e. $T:\M_2\to\M_2$ is a doubly stochastic qubit quantum channel. Such a quantum channel can be represented as an affine transformation on $\R^3$, the so-called Bloch sphere representation~\cite{nielsen2000quantum}. In this representation quantum states $\rho\in\D_2$ are identified with vectors $x\in \R^3$ by
\begin{align*}
\rho = \frac{\one_2 + \sum^3_{i=1} x_i\sigma_i}{2}
\end{align*}  
where we used the Pauli-matrices, i.e.
\begin{align*}
\sigma_1 = \begin{pmatrix} 0 & 1 \\ 1 & 0\end{pmatrix} ,~ \sigma_2 = \begin{pmatrix} 0 & -i \\ i & 0\end{pmatrix},~ \sigma_3 = \begin{pmatrix} 1 & 0 \\ 0 & -1\end{pmatrix}.
\end{align*}
In this representation a doubly stochastic quantum channel $T:\M_2\ra\M_2$ corresponds to a matrix $\hat{T}\in\M_3$ acting on the corresponding vectors in $\R^3$. Note that the adjoint channel $T^*$ with respect to the Hilbert-Schmidt scalar product corresponds to the transposed matrix $\widehat{(T^*)} = (\hat{T})^T$.  

We now show how to express the LS constants of a doubly stochastic qubit Liouvillian of the form $\liou = T -\id_2$ in terms of the matrix $\hat{T}$ representing the action of the quantum channel $T$.

%
%
%
%
%

\begin{thm}\label{alpha2qubit}
 Let $\liou:\M_2\to\M_2$ be a doubly stochastic and primitive Liouvillian of the form 
 $\liou=T-\text{id}_2$, where $T:\M_2\ra\M_2$ is a doubly stochastic quantum channel. Then we have 
\begin{align*}
\alpha_2\lb \liou\rb =1-\Big{\|}\frac{\hat{T} + \hat{T}^T}{2}\Big{\|}
\end{align*} 
where $\hat{T}\in\M_3$ is the matrix representing $T$ on the Bloch sphere.
\end{thm}
\begin{proof}
Note that we may replace $\liou$ by the symmetrized Liouvillian $\frac{\liou+\liou^*}{2}$ as we are computing the LS-2 constant. Thus, we may consider $\frac{T+T^*}{2}$, which is represented by a symmetric matrix, instead of $T$.

Diaconis and Saloff-Coste showed that the LS-2 constant of a Markov kernel $M-\one_d$ on a two-dimensional state space is given by $2M_{1,2}$
~\cite[Corollary A.3]{diaconis1996}\footnote{Note that our definition of the 2-entropy in Definition \ref{defn:LS2} is half of the corresponding function in \cite{diaconis1996}. This
 explains the difference by a factor of 2 between our results.}. For $U\in\U_d$ consider $M_U$, see \eqref{equ:Markovkernel}. Using \eqref{equ:classLS2} and the aforementioned result we get:
 \begin{equation}\label{c2eq1}
\alpha_2(\liou) = \inf_{U\in\U_d}\alpha_2^{(c)}\lb M_U -\one_d\rb =2 \inf_{U\in\U_d}\bra{1}U^\dagger \lb\frac{T+T^*}{2}\rb\lb U\proj{0}{0}U^\dagger\rb U\ket{1}
 \end{equation}
Let $x$ be the vector corresponding to $U\ket{0}$ on the Bloch sphere. Then $-x$ is the vector corresponding to $U\ket{1}$. By changing to the Bloch sphere representation we get:
\begin{equation}\label{c2eq2}
\bra{1}U^\dagger \lb\frac{T+T^*}{2}\rb\lb U\proj{0}{0}U^\dagger\rb U\ket{1} =\frac{1}{2}-\frac{\scalar{x}{\lb\frac{\hat{T}+\hat{T}^T}{2}\rb x}}{2}.
\end{equation}
Taking the infimum over $\U_d$ in \eqref{c2eq1} corresponds to taking the infimum over all
unit vectors in $\mathbb{R}^3$ in \eqref{c2eq2}. This gives
\begin{align*}
 \alpha_2\lb \liou\rb =1-\sup\limits_{x\in\R^3:\|x\|=1}\scalar{x}{\lb\frac{\hat{T}+\hat{T}^T}{2}\rb x}=1-\Big{\|}\lb\frac{\hat{T}+\hat{T}^T}{2}\rb\Big{\|}.
\end{align*}

\end{proof}

Similarly we can also compute the LS-1 constant:

\begin{thm}
Let $\liou:\M_2\to\M_2$ be a doubly stochastic and primitive Liouvillian of the form $\liou=T-\text{id}$, where $T:\M_2\ra\M_2$ is a quantum channel. Then, 
\begin{align*}
\alpha_1\lb \liou\rb =1-\sup\limits_{x\in\R^3:\|x\|=1}|\scalar{x}{\hat{T}x}|  ~, 
\end{align*}
where $\hat{T}\in\M_3$ is the matrix representing the action of $T$ on the Bloch sphere. 
\label{thm:alpha1QubitT}
\end{thm}
\begin{proof}
For fixed $U\in\U_2$ we define $p:= (M_U)_{1,2}$, see \eqref{equ:Markovkernel}. As $M_U$ is doubly stochastic for a doubly stochastic quantum channel $T$ we may write the Markov kernel as:
\begin{align}
 M_U - \one_2 =\left( \begin{array}{ccc}
-p & p \\
p & -p\\
\end{array} \right)=: 2p Q,
\label{equ:RwalkMKern}
\end{align}
where $Q := \begin{pmatrix}
-\frac{1}{2} & \frac{1}{2} \\ \frac{1}{2} & -\frac{1}{2}
\end{pmatrix}$ denotes the kernel of a random walk on the complete graph with two vertices and uniform distribution.
In \cite{diaconis1996} the classical LS-1 constant is defined as
\begin{align*}
\alpha^{(c)}_1(M_U-\one_2) = \inf_{s\in\R_+^2} -\frac{1}{2}\frac{\scalar{2pQ \log\lb s \rb }{s}}{\Ent_1\lb X\rb } = 2p \alpha^{(c)}_1(Q)
\end{align*}
where we used \eqref{equ:RwalkMKern}. It has been shown in \cite{bobkov2006modified} that $\alpha^{(c)}_1(Q) = 1$. By \eqref{equ:classLS1} we can now compute
\begin{align*}
\alpha_1(\liou) = \inf\limits_{U\in\U_d} \alpha^{(c)}_1(M_U-\one_d) = \inf\limits_{U\in\U_d} 2\bra{1}U^\dagger T\lb U\proj{0}{0}U^\dagger\rb U\ket{1}. 
\end{align*}
Changing to the Bloch sphere representation as in the proof of Theorem \ref{alpha2qubit} finishes the proof. 
 \end{proof}
 
Using Theorem \ref{entdiff} we obtain the following entropy production estimate from Theorem \ref{thm:alpha1QubitT}:

\begin{cor}[Entropy production of qubit Liouvillians of the form $\Tm-\id$]\hfill\\
Let $\liou:\M_2\to\M_2$ be a doubly stochastic and primitive Liouvillian of the form $\liou=T-\text{id}$ for a quantum channel represented by $\hat{T}\in\M_3$ in the Bloch sphere picture with $r=\sup\limits_{x\in\R^3:\|x\|=1}|\scalar{x}{\hat{T}x}|$. Then
\begin{align*}
S(T_t(\rho))-S(\rho)\geq(1-e^{-2(1-r)t})(log(d)-S(\rho))
\end{align*}
for any $\rho\in\D_d$ and where $T_t=e^{\liou t}$ denotes the semigroup generated by $\liou$.
\end{cor}

The results from the previous section show that for doubly stochastic, primitive and reversible qubit Liouvillians of the form $\Lm = T-\id$ we have $\alpha_1\lb\liou\rb =\alpha_2\lb\liou\rb$. However, in general we might have $\alpha_1\lb\liou\rb\gg\alpha_2\lb\liou\rb$ even for reversible Liouvillians of this type. This is demonstrated for instance by the depolarizing Liouvillian $\liou(X)=\tr(X)\frac{\one}{d}-X$ where
\begin{align*}
 \alpha_2(\liou)=\frac{2(1-\frac{2}{d})}{\log(d-1)}\ra 0\text{ as }d\ra\infty\hspace{0.3cm},\text{ but  }\hspace{0.3cm} \alpha_1(\liou)\geq\frac{1}{2}.
\end{align*}
Therefore, methods based on hypercontractivity may lead to entropy production estimates that are far from optimal, as the optimal constant is described by the LS-1 constant and the separation between LS-2 and LS-1 can be arbitrarily large. For the depolarizing channels in any dimension the authors succeeded in computing the exact LS-1 constant and thus the optimal entropy production in~\cite{otherPaper}.

\section{Discrete LS inequalities for doubly stochastic channels}
\label{sec:4}

In this section we show how LS inequalities may be used to derive an entropy production estimate for doubly stochastic quantum channels in discrete time. We build upon and generalize some results of \cite{miclo}. As for continuous time-evolutions, we say that a doubly stochastic quantum channel $T:\M_d\to\M_d$ is primitive iff $\lim\limits_{n\to\infty}T^n(\rho)=\frac{\one_d}{d}$ for any $\rho\in\D_d$. We will need the following characterization of primitive channels:

\begin{thm}[\cite{quantumwielandt}]\label{charprimitive}
Let $T:\M_d\to\M_d$ be a doubly stochastic quantum channel. The following are equivalent:
\begin{enumerate}
 \item T is primitive.
 \item There exists $n\in\mathbb{N}$ such that we have $T^n\lb \rho\rb >0$ for any $\rho\in\D_d$.
 \item $T$ has only one eigenvalue of magnitude $1$ counting multiplicities.
\end{enumerate}

\end{thm}


We define the discrete LS constant of a quantum channel $T:\M_d\ra\M_d$ by the usual LS-2 constant, see Definition \ref{defn:LS2}, of a continuous semigroup associated to $T$.

\begin{defn}[Discrete LS constant]
 Let $T:\M_d\to\M_d$ be a doubly stochastic quantum channel such that $T^*T$ is primitive. 
 Define the discrete LS constant of
 $T$ as
 \begin{equation}
  \alpha_D\lb T\rb :=\frac{1}{2}\alpha_2\lb T^*T-\text{id}_d\rb 
 \end{equation}
 \label{defn:discLSconst}

\end{defn}

This definition is motivated by the following entropy production estimate:

\begin{thm}\label{timproveddata}
  Let $T:\M_d\to\M_d$ be a doubly stochastic quantum channel such that $T^*T$ is primitive. Then for any $\rho\in\D_d$
  we have:
\begin{equation}\label{improveddata}
   D\lb T\lb \rho\rb \Big{\|}\frac{\one_d}{d}\rb\leq\lb 1-\alpha_D\lb T\rb \rb D\lb\rho\Big{\|}\frac{\one_d}{d}\rb
  \end{equation}      
  
\end{thm}

Note that this is equivalent to the entropy production estimate:
\begin{align*}
  S(T(\rho)) - S(\rho) \geq \alpha_D(T)(\log(d)-S(\rho))
\end{align*} 

 \begin{proof}
 
 Suppose first that $\rho\in\D_d^+$ and define $X=d\rho$.
 Our goal is to show the following inequality:
 \begin{equation}\label{ineq1}
 D\lb T\lb \rho\rb \Big{\|}\frac{\one_d}{d}\rb\leq  D\lb\rho\Big{\|}\frac{\one_d}{d}\rb-\Em^2_{T^*T-\text{id}}\lb X^{\frac{1}{2}}\rb.
\end{equation}
For that we will use the theory of discrete LS inequalities for classical Markov chains introduced in~\cite{miclo}. Let $\{\ket{a_i}\}_{1\leq i\leq d}$ and $\{\ket{b_j}\}_{1\leq j\leq d}$ be orthonormal bases of $\C^d$ consisting of eigenvectors of $X$ and $T\lb X\rb $, respectively. 
As $T$ is a doubly stochastic quantum channel the matrix $P\in\M_d$ defined as $\lb P\rb _{i,j} :=\bra{b_i}T\lb \proj{a_j}{a_j}\rb \ket{b_i}$ is doubly-stochastic. In the following denote by $s(X)\in\R_+^d$ the vector of eigenvalues of $X$ decreasingly ordered. 

Using that as $\one_d$ commutes with any $d\times d$-matrix we have $D(\rho\|\frac{\one_d}{d}) = D^{(c)}(s(\rho)\|\pi_d)$ and $D(T(\rho)\|\frac{\one_d}{d}) = D^{(c)}(Ps(\rho)\|\pi_d)$ 
 where $D^{(c)}\lb\cdot\|\cdot\rb$ denotes the classical relative entropy and $\pi_d$ the $d$-dimensional uniform distribution. 
Then with \cite[Equation 8]{miclo}
\footnote{Note that in \cite{miclo} the evolution of a probability distribution is given by a left multiplication with the transition matrix $P$, 
while we work with right multiplication. This explains why the order of $P$ and $P^{T}$ is reversed.} one can easily show:
\begin{equation}
 D\lb T\lb \rho\rb \Big{\|}\frac{\one_d}{d}\rb\leq D\lb\rho\Big{\|}\frac{\one_d}{d}\rb 
 -\frac{1}{d}\sum\limits_{i,j=1}^d\sqrt{s\lb X\rb }_i\lb \sqrt{s\lb X\rb }_i-\sqrt{s\lb X\rb }_j\rb \lb P^TP\rb _{i,j}.
 \label{equ:blablablup}
\end{equation}

Let $V\in\U_d$ be a unitary operator such that $\lbr VT\lb X\rb V^\dagger,T\lb X^{\frac{1}{2}}\rb\rbr = 0$, i.e. both operators have the same eigenvectors $\{\ket{c_k}\}_{1\leq k\leq d}$.
Now define the doubly-stochastic matrix $\lb Q\rb _{i,k}=\bra{c_k}T\lb \ketbra{a_i}\rb \ket{c_k}$. As $Q$ is doubly-stochastic so is $Q^{T}Q$ and we have:
\begin{equation}\label{eqtr}
\sum\limits_{i,k=1}^ds\lb X\rb_i\lb Q^{T}Q\rb_{i,k}=\text{tr}\lb X\rb.
\end{equation}
By construction, $Qs\left(X^{\frac{1}{2}}\right)=s\lb T\lb X^{\frac{1}{2}}\rb\rb$ and so:
\begin{equation}\label{eqtrroot}
\sum\limits_{i,k=1}^ds\lb X^{\frac{1}{2}}\rb_is\lb X^{\frac{1}{2}}\rb_k\lb Q^{T}Q\rb_{i,k}=\scalar{Qs\lb X^{\frac{1}{2}}\rb}{Qs\lb X^{\frac{1}{2}}\rb}=\tr\left[T\lb X^{\frac{1}{2}}\rb^2\right] 
\end{equation}
Using unitary invariance of the relative entropy, \eqref{equ:blablablup} and equations (\ref{eqtr}) and (\ref{eqtrroot}) we have that:
\begin{align}\label{uptounitary}
 D\lb T\lb \rho\rb \Big{\|}\frac{\one_d}{d}\rb &= D\lb VT\lb \rho\rb V^\dagger \Big{\|}\frac{\one_d}{d}\rb \\
 &\leq D\lb \rho\Big{\|}\frac{\one_d}{d}\rb -\frac{1}{d}\sum\limits_{i,k=1}^d\sqrt{s\lb X\rb }_i\lb \sqrt{s\lb X\rb }_i-\sqrt{s\lb X\rb }_k\rb \lb Q^{T}Q\rb _{i,k}\\
 & = D\lb \rho\Big{\|}\frac{\one_d}{d}\rb - \frac{1}{d}\tr\lb X-T\lb X^{\frac{1}{2}}\rb ^2\rb\\
 &=  D\lb \rho\Big{\|}\frac{\one_d}{d}\rb - \Em^2_{T^*T-\text{id}}\lb X^{\frac{1}{2}}\rb. 
\end{align}

By Definition \ref{defn:discLSconst} of the discrete LS-2 constant and by definition of $X$ we have
\begin{align*}
D\lb T\lb \rho\rb \Big{\|}\frac{\one_d}{d}\rb \leq D\lb \rho\Big{\|}\frac{\one_d}{d}\rb -\Em^2_{T^*T-\id}\lb \lb d\rho\rb ^{\frac{1}{2}}\rb \leq \lb 1-\alpha_D\lb T\rb \rb  D\lb \rho\Big{\|}\frac{\one_d}{d}\rb 
\end{align*}
for full rank $\rho$. The inequality follows for all states by a continuity argument.

\end{proof}

Note that if $T^*T$ is primitive, then $\alpha_D(T)$ is strictly positive by the lower bound given in Theorem \ref{relationconstants} as primitivity implies the existence of a positive spectral gap by Theorem \ref{charprimitive}.

If the channel $T$ is normal, i.e. $T^*T=TT^*$, it was shown in \cite[Proposition 13]{ergodicchiribella} that 
$T^*T$ being primitive is also a necessary condition for $T$ being primitive. 
However, as being a strict contraction w.r.t. the relative entropy is not a necessary condition for 
primitivity, it can't hold that $\alpha_D(T)>0$ for all primitive channels $T$.
We now show that the assumption of $T$ being primitive is sufficient to ensure that $\alpha_D\lb T^n\rb>0$
for some $n\in\N$
and also derive another characterization of primitive, doubly stochastic channels:
\begin{thm}
 Let $T:\M_d\to\M_d$ be a doubly stochastic quantum channel. The following are equivalent:
 \begin{enumerate}
  \item $T$ is primitive
  \item $\exists n\in\mathbb{N}$ such that $\lb T^*\rb ^nT^n$ is primitive
 \end{enumerate}

\end{thm}
\begin{proof}
 If $T$ is primitive, then by Theorem \ref{charprimitive} there exists $n\in\mathbb{N}$
 such that $T^n\lb \rho\rb >0$ for all $\rho\in\D_d$. As $T$ is doubly stochastic, also $\lb T^*\rb ^nT^n\lb \rho\rb >0$
 holds, which implies that $\lb T^*\rb ^nT^n$ is primitive.
 On the other hand, if $\lb T^*\rb ^nT^n$ is primitive for some $n\in\mathbb{N}$,
 then $\alpha_D\lb T^n\rb >0$. By Theorem \ref{timproveddata} we have $\lim\limits_{k\to\infty}T^{kn}\lb \rho\rb =\frac{\one_d}{d}$. As $\frac{\one_d}{d}$ is a fixed point the convergence for a subsequence implies the convergence of the sequence.
\end{proof}

We were not able to determine if inequality (\ref{improveddata}) is tight, i.e. whether there is a primitive doubly stochastic channel $T:\M_d\to\M_d$ with
\begin{align*}
\sup_{\rho_\in\D_d,\rho\not=\frac{\one}{d}}\frac{D\lb T\lb \rho\rb\|\frac{\one}{d}\rb}{D\lb \rho\|\frac{\one}{d}\rb}=1-\alpha_D(T) .
\end{align*}
However, it is clear that the best constant in the equation \eqref{improveddata} is not always given by the discrete LS constant. Consider for instance the completely depolarizing channel $T\lb \rho\rb :=\tr\lb \rho\rb \frac{\one_d}{d}$. Then, $\alpha_D\lb T^*T-\id\rb =\frac{\lb1-\frac{2}{d}\rb}{\log\lb d-1\rb }<1$~\cite[Corollary 27]{Kastoryanosob}.

Now we want to determine a bound on the discrete LS constant. For this we need:
\begin{thm}\label{monotonicitysobolev}
 Let $T:\M_d\to\M_d$ be a doubly stochastic, primitive quantum channel.
 Then the function $k\mapsto\alpha_2\lb \lb T^*\rb^kT^k-\id_d\rb $ is monotone increasing.
\end{thm}
\begin{proof}
 As $T$ is doubly stochastic, the operator Schwarz inequality\cite[Proposition 3.3]{paulsen2002completely} implies that for $X\in\M_d^+$ 
 \begin{equation}
T^{k+1}\lb X\rb ^2\leq T\lb T^k\lb X\rb ^2\rb.   
 \end{equation}
 As $T$ is trace-preserving we get
 \begin{equation}\label{tcontrac}
  \|T^{k+1}\lb X\rb \|_{2,\frac{\one_d}{d}}^2=\frac{1}{d}\tr\lb T^{k+1}\lb X\rb ^2\rb \leq\frac{1}{d}\tr\lb T^k\lb X\rb ^2\rb =\|T^{k}\lb X\rb \|_{2,\frac{\one_d}{d}}^2
 \end{equation}
 where again $\|X\|_{2,\sigma}^2=\frac{1}{d}\tr\left(X^2\right)$.
 Observe that
\begin{align*}
\Em^2_{(T^*)^k T^k-\id_d}\lb X\rb =-\frac{1}{d}\tr\lb\lb(T^*)^k T^k\rb\lb X\rb X-X^2\rb=\|T^k\lb X\rb\|_{2,\frac{\one}{d}}^2 - \|X\|_{2,\frac{\one}{d}}^2.
\end{align*}

Then we have by \eqref{tcontrac} that $\Em^2_{(T^*)^{k+1} T^{k+1}-\id_d}\geq\Em^2_{\lb(T^*)^k T^k\rb-\id_d}$ , which implies the claim by the variational definition of the LS-2 constant.

\end{proof}

The next corollary shows that the discrete LS constant becomes less useful if the dimension $d$ of the system grows large.

\begin{cor}\label{boundsdiscrete}

 Let $T:\M_d\to\M_d$ be a doubly stochastic quantum channel s.t. $T^*T$ is primitive. Then
 \begin{equation}
  \min\{\frac{\lambda}{2},\frac{\lb 1-\frac{2}{d}\rb }{\log\lb d-1\rb }\}\geq\alpha_D\lb T\rb \geq\lambda\frac{\lb 1-\frac{2}{d}\rb }{\log\lb d-1\rb },
 \end{equation}
where $\lambda$ is the spectral gap of $T^*T-\id$. Again we have $\frac{2\lb 1-\frac{2}{d}\rb}{\log(d-1)}=1$  for $d=2$ by continuity.
\end{cor}
\begin{proof}
 By Theorem \ref{monotonicitysobolev}, for $k\in\mathbb{N}$ we have $\alpha_D\lb T\rb \leq\alpha_D\lb T^k\rb $ and so
\begin{equation}
  \alpha_D\lb T\rb \leq\lim\inf\limits_{k\to\infty}\alpha_D\lb T^k\rb
\end{equation}
By Theorem \ref{lowerboundtensor}:
\begin{align*}
\lim\inf\limits_{k\to\infty}\alpha_D\lb T^k\rb\leq\lim\limits_{k\to\infty}\|\lb T^{*}\rb ^kT^{k}-\id_d\|\frac{\lb 1-\frac{2}{d}\rb }{\log\lb d-1\rb }
\end{align*}
For primitive channels $\lim\limits_{k\to\infty}T^k=T_{\infty}$, where $T_\infty\lb X\rb=\tr\lb X\rb \frac{\one_d}{d}$. 
By the continuity of norms, multiplication and conjugation of linear operators and $\|\id_d-T_{\infty}\|=1$:
\begin{align*}
\alpha_D\lb T\rb \leq \frac{\lb 1-\frac{2}{d}\rb }{\log\lb d-1\rb } 
\end{align*}
 
 The lower bound follows from \cite[Corollary 27]{Kastoryanosob} and the upper bound in terms of
 the spectral gap from Theorem \ref{relationconstants}.
\end{proof}

With Theorem \ref{timproveddata} we immediately get the following entropy production estimate for a doubly stochastic quantum channel $T:\M_d\ra\M_d$
\begin{equation}\label{entproddiscspectral}
S(T(\rho)) - S(\rho)\geq \lambda\frac{\lb 1-\frac{2}{d}\rb }{\log\lb d-1\rb} (\log(d) - S(\rho))
\end{equation}
for any $\rho\in\D_d$, where $\lambda$ denotes the spectral gap of $T^*T-\id$.
\begin{rem}\label{compotherbounds}
The bound given in Equation \eqref{entproddiscspectral} is similar to the one in \cite[Lemma 2.1]{streater}, given by
\begin{align*}
S(T(\rho)) - S(\rho)\geq \frac{\lambda}{2}\|\rho-\frac{\one}{d}\|^2_2,
\end{align*}
where $\lambda$ is again the spectral gap of $T^*T-\id$. However, if $\alpha_D(T)$ and $\lambda$ are of the same order of 
magnitude, then \eqref{improveddata} gives an improvement of order $\log(d)$.
Similar holds for the bound in \cite{raginsky2002entropy}.
\end{rem}

Given the usefulness of the hypercontractive characterization of the LS-2 constant in continuous time, it is natural to ask whether we have a similar
characterization of the discrete LS constant. 
We now show that the discrete LS constant implies hypercontractivity of the channel,
but first
we will prove some technical lemmas. The following lemma is the generalization of \cite[Lemma 3]{miclo} to the quantum case.
\begin{lem}\label{lemma1}
For $X\in\M^+_d$ and $q\geq2$:
 \begin{equation}
  \|X\|_{q,\frac{\one_d}{d}}-\|X\|_{2,\frac{\one_d}{d}}\leq\frac{q-2}{q}\|X\|_{q,\frac{\one_d}{d}}^{1-q} \mathrm{Ent}_2\lb X^\frac{q}{2}\rb 
 \end{equation}

\end{lem}
\begin{proof}
 Working in the eigenbasis of $X$, the proof is identical to \cite[Lemma 3]{miclo}.
\end{proof}
\begin{lem}\label{lemma2}
 Let $T:\M_d\to\M_d$ be a doubly stochastic quantum channel, $X\in\M_d^+$ and $q\geq2$. Then:
 \begin{equation}\label{claimlemma2}
  \|T\lb X\rb \|_{q,\frac{\one_d}{d}}^q-\|X\|_{q,\frac{\one_d}{d}}^q\leq-\Em^2_{T^*T-\id}\lb X^{\frac{q}{2}}\rb 
 \end{equation}

\end{lem}
\begin{proof}
As the r.h.s. of \ref{claimlemma2} is equal to
\begin{align}
\frac{1}{d}\left[\tr\lb T\lb X^{\frac{q}{2}}\rb^2\rb-\tr\lb X^q\rb\right]
\end{align}
and the l.h.s. is equal to
\begin{align}
\frac{1}{d}\left[\tr\lb T\lb X\rb ^q\rb-\tr\lb X^q\rb\right],  
\end{align}
the claim is equivalent to $\tr\lb T\lb X\rb ^q\rb \leq\tr\lb T\lb X^\frac{q}{2}\rb ^2\rb $.
 As $T$ is a doubly stochastic channel and $q\geq2$~\cite{schwarzconvex}:
 \begin{equation}\label{opineq}
T\lb X\rb \leq T\lb X^\frac{q}{2}\rb ^{\frac{2}{q}}
 \end{equation} 
 Both $T\lb X\rb $ and $T\lb X^\frac{q}{2}\rb ^{\frac{2}{q}}$ are positive operators, which implies by 
 Weyl's monotonicity theorem~\cite[Corollary III.2.3]{bhatia1997matrix}: 
\begin{equation}
s\lb T\lb X\rb \rb _i\leq s\lb T\lb X^\frac{q}{2}\rb ^\frac{2}{q}\rb _i  
\end{equation}
As $\tr\lb T\lb X\rb ^q\rb =\sum\limits_{i=1}^d s\lb T\lb X\rb \rb _i^q$, we have:
\begin{equation}
\tr\lb T\lb X\rb ^q\rb \leq\sum\limits_{i=1}^d s\lb T\lb X^\frac{q}{2}\rb \rb _i^2=\tr\lb T\lb X^\frac{q}{2}\rb ^2\rb  
\end{equation}
and the claim follows.
\end{proof}

\begin{thm}[Discrete hypercontractivity]
Let $T:\M_d\to\M_d$ be a doubly stochastic quantum channel s.t. $T^*T$ is primitive. Then 
$\|T\|_{2\to q,\frac{\one_d}{d}}\leq1$ for $q\leq2+2\alpha_D\lb T\rb $
\label{thm:DiscreteHyper}
\end{thm}

\begin{proof}

For $X\in\M_d^+$, we have by Lemma (\ref{lemma1}):
\begin{equation}\label{ineqlemma1}
 \|X\|_{q,\frac{\one_d}{d}}-\|X\|_{2,\frac{\one_d}{d}}\leq\frac{q-2}{q}\|X\|_{q,\frac{\one_d}{d}}^{1-q}\Ent_2\lb X^\frac{q}{2}\rb 
\end{equation}

As the function $g\lb x\rb =x^{\frac{1}{q}}$ is concave on $\R_+$,  we have the following
inequality for $a,b\in\R_+$:
\begin{equation}
 a^{\frac{1}{q}}-b^{\frac{1}{q}}\leq\frac{1}{q}\lb b^{\frac{1}{q}-1}\lb a-b\rb \rb
\end{equation}
Plugging in $a=\|T\lb X\rb \|_{q,\frac{\one_d}{d}}^q$ and $b=\|X\|_{q,\frac{\one_d}{d}}^q$ we get:
\begin{equation}\label{concavnorm}
 \|T\lb X\rb \|_{q,\frac{\one_d}{d}}
  -\|X\|_{q,\frac{\one_d}{d}}\leq \frac{1}{q}\|X\|_{q,\frac{\one_d}{d}}^{1-q}\lb\|T\lb X\rb \|_{q,\frac{\one_d}{d}}^q-\|X\|_{q,\frac{\one_d}{d}}^q\rb
  \end{equation}
 Summing inequalities (\ref{ineqlemma1}) and (\ref{concavnorm}) and applying Lemma (\ref{lemma2}), 
 we finally get:
 \begin{equation}\label{ddifference}
   \|T\lb X\rb \|_{q,\frac{\one_d}{d}}-\|X\|_{2,\frac{\one_d}{d}}\leq \frac{1}{q}\|X\|_{q,\frac{\one_d}{d}}^{1-q}
   \lb(q-2)\Ent_2\lb X^{\frac{q}{2}}\rb -\Em^2\lb X^{\frac{q}{2}}\rb \rb
 \end{equation}
By the definition of the discrete LS inequality, if $q\leq2+2\alpha_D\lb T\rb $
the right-hand side in (\ref{lemma2}) is negative, which implies for $X\in\M_d^+$:
\begin{equation}\label{hyperdiscineq}
  \frac{\|T\lb X\rb \|_{q,\frac{\one_d}{d}}}{\|X\|_{2,\frac{\one_d}{d}}}\leq1
\end{equation}
As shown in \cite{audenaertptoq,watrousptoq}, we may restrict ourselves to positive semi-definite operators
when considering the $2\to q$ norm of completely positive maps. 
By continuity, inequality (\ref{hyperdiscineq}) is also valid for all positive semi-definite operators, implying
$\|T\|_{2\to q,\frac{\one_d}{d}}\leq1$.

\end{proof}

It is not to be expected that the other direction holds, at least in a naive way. That is, that a bound of the form:
\begin{align*}
\|T\|_{2\to q,\frac{\one_d}{d}}\leq1 
\end{align*}
for some $q>2$ and a doubly stochastic quantum channel $T:\M_d\to\M_d$ gives a lower-bound on $\alpha_D(T)$ that is independent of the dimension $d$, as we have for continuous time. 

To see why this is the case, consider a doubly stochastic qubit channel $T:\M_2\to\M_2$. As we have mentioned before, we have~\cite{king}:
\begin{align*}
 \|T^{\otimes n}\|_{2\to q,\frac{\one_{d^n}}{d^n}}=\|T\|_{2\to q,\frac{\one_d}{d}}^n
\end{align*}
If hypercontractivity of the channel would imply a lower-bound on $\alpha_D(T)$ that is independent of the dimension, we would then obtain a lower-bound
strictly positive
for $T^{\otimes n}$ for all $n\in\N$. But it follows from Corollary \ref{boundsdiscrete} that $\lim\limits_{n\to\infty}\alpha_D\lb T^{\otimes n}\rb=0$.

\section{Applications}
\label{sec:6}

\subsection{Upper bounds on unitary quantum subdivision capacities}
\label{sec:subd}

In \cite{subdivisioncapacities} some of the authors introduced quantum capacities for continuous Markovian time-evolutions. These capacities are similar to the usual quantum capacity, but in addition to applying encoding and decoding operations in the beginning and end of the protocol additional operations may be applied in intermediate steps. Here we will only consider the case where these additional operations are unitary quantum channels. The precise definition of this unitary quantum subdivision capacity is:
 
\begin{defn}[Unitary quantum subdivision capacity $\mathcal{Q}_\chanSet$\cite{subdivisioncapacities}]\hfill
\label{defn:SubdivisionCap}

The $\mathfrak{U}$-\textbf{quantum subdivision capacity} of a  Liouvillian $\Lm:\M_{d}\ra \M_{d}$ at a time $t\in\R_+$ is defined as 
\begin{align*}
\mathcal{Q}_\mathfrak{U}\lb t\Lm\rb := \sup\lset R\in \R_+ :\text{ R achievable rate}\rset
\end{align*}
where a rate $R\in\R_+$ is called achievable if there exist sequences $\lb n_\nu\rb^\infty_{\nu=1},\lb m_\nu\rb^\infty_{\nu=1}$ such that $R = \limsup_{\nu\ra \infty}\frac{n_\nu}{m_\nu}$ and we have 
\begin{align}
\inf_{k,E,D,U_1,\ldots,U_k} \left\Vert \id^{\otimes n_\nu}_2 - D\circ \prod^k_{l=1} \lb U_l\circ T^{\otimes m_\nu}_{\frac{t}{k}}\rb\circ E\right\Vert_{\diamond} \ra 0\hspace*{0.3cm}\text{as }\nu\ra \infty.
\label{equ:CodingSubdivisionCap}
\end{align}
The latter infimum is over the number of subdivisions $k\in\N$ for which the channels $T_{\frac{t}{k}} := e^{\frac{t}{k}\Lm}$ are defined, arbitrary encoding and decoding quantum channels $E:\M^{\otimes n_\nu}_{2}\ra \M^{\otimes m_\nu}_{d}$ and $D:\M^{\otimes m_\nu}_{d}\ra \M^{\otimes n_\nu}_{2}$ and unitary channels $U_l\in \mathfrak{U}$ from the chosen subset.
\end{defn}

The unitary quantum subdivision capacity quantifies the highest possible rate of information storage in a quantum memory, when unitary gates may be applied to protect the information against the noise. Using our Corollary \ref{cor:TSBoundAlpha} we obtain the following upper bound on $Q_\mathfrak{U}$:

\begin{thm}(Upper bound for doubly stochastic, primitive Liouvillians)\hfill \\
\label{thm:UnitarySubDCapDepolBound}
Let $\liou:\M_d\to\M_d$ be a doubly stochastic and primitive Liouvillian with spectral gap
 $\lambda$. Then 
\begin{align*}
\mathcal{Q}_\mathfrak{U}\lb t\Lm\rb\leq e^{-2t\alpha\lb d\rb}\log(d)\hspace*{0.3cm}.
\end{align*}
with $\alpha(d)=\frac{\lambda\lb1-2d^{-2}\rb}{\log(3)\log(d^2-1)+2\lb1-2d^{-2}\rb}$.
\end{thm} 
\begin{proof}
Using Corollary \ref{cor:EntrProdTens} gives 
\begin{align*}
S(T^{\otimes n}_t\lb\rho\rb) &\geq (1-e^{-2t\alpha\lb d\rb})\log(d^n) + e^{-2t\alpha\lb d\rb} S\lb\rho\rb \\
&\geq (1-e^{-2t\alpha\lb d\rb})\log(d^n)
\end{align*}
where $T_t = e^{t\Lm}$ and $n\in\N$. The rest of the proof follows the lines of the proof in~\cite[Theorem 5.1]{subdivisioncapacities}.

\end{proof}

The above theorem shows that in a quantum memory affected by a doubly stochastic, primitive and self-adjoint noise Liouvillian the storage rate is exponentially small in time, when only unitary correction operations are allowed. 
This result is similar in flavor to results by Ben-Or et al.~\cite{ben2013quantum,benor}.

 \subsection{Entropy production for random Pauli channels}
As an application of the discrete LS inequalities from Section \ref{sec:4}, we derive an entropy production estimate for random Pauli channels.

\begin{defn}[Random Pauli channel]
 A channel $T:\M_2\to\M_2$ is said to be a random Pauli channel
 if it can be written as 
 \begin{align*}
 T\lb \rho\rb =p_1\sigma_1\rho \sigma_1+p_2\sigma_2\rho \sigma_2+p_3\sigma_3\rho \sigma_3+\lb 1-p_1-p_2-p_3\rb \rho,
 \end{align*}
 for a probability distribution $\lb p_1,p_2,p_3,1-p_1-p_2-p_3\rb$. Here $\sigma_1,\sigma_2,\sigma_3$ are the Pauli matrices.
\end{defn}

First we use the results from Section \ref{sec:5} to prove: 

\begin{thm}\label{sobolevpauli}
 Let $T:\M_2\ra\M_2$ be a random Pauli channel and $\liou:\M_2\to\M_2$ given by $\liou=T-\id_2$.
 Then the LS-2 constant (Definition \ref{defn:LS2}) can be computed as
 \begin{align*}
  \alpha_2\lb \liou\rb =2\min\{p_1+p_2,p_2+p_3,p_3+p_1\} .
 \end{align*}

\end{thm}
\begin{proof}
Note that we have $\alpha_2\lb\liou\rb = 0$ if $\liou$ is not primitive. Therefore, we only have to show the claim for $\liou$ primitive. As $T$ is reversible, Corollary \ref{cor:Qubit} implies $\alpha_1\lb \liou\rb =\alpha_2\lb \liou\rb =\lambda\lb \liou\rb $.
One can check that the spectrum of a random Pauli channel is given by:
\begin{align*}
 \{1,1-2\lb p_1+p_2\rb ,1-2\lb p_3+p_2\rb ,1-2\lb p_1+p_3\rb \}
\end{align*}
and thus the spectral gap of $\liou$ is given by $2\min\{p_1+p_2,p_2+p_3,p_3+p_1\}$.
\end{proof}

\begin{thm}

 Let $T:\M_2\to\M_2$ be a random Pauli channel.
  Define
 \begin{align*}
  p=2\min\{p_1p_2+p_1p_3,p_2p_1+p_2p_3,p_3p_1+p_3p_2\}.
 \end{align*}

 Then $S(T(\rho)) - S(\rho) \geq p(\log(d)-S(\rho))$.

\end{thm}
\begin{proof}
 It is easy to check that if $T$ is a random Pauli channel,
 \[
 T^*T\lb \rho\rb =q_1\sigma_1\rho \sigma_1+q_2\sigma_2\rho \sigma_2+q_3\sigma_3\rho \sigma_3+\lb 1-q_1-q_2-q_3\rb \rho 
 \]
is a random Pauli channel as well with 
 $q_1=2p_2p_3,q_2=2q_1q_3$ and $q_3=2q_1q_2$.
By Theorem \ref{sobolevpauli}, $\alpha_D\lb T\rb =p$ and the claim follows from Theorem \ref{timproveddata}.

\end{proof}

\section{Conclusion}
We have extended the use of group theoretic techniques to study LS inequalities for doubly stochastic, primitive Markovian time-evolutions. These bounds lead to entropy-production estimates for tensor-powers of this kind of semigroups, which are independent of the number of tensor-powers. We applied these estimates to derive upper bounds on quantum subdivision capacities. For discrete doubly stochastic quantum channels we generalized discrete LS inequalities to the quantum case.
 
There are some directions of possible future research which should be emphasized. It would be desirable to try our group theoretic approach with other relevant semigroups and generalize
it to semigroups that are not doubly stochastic. Another concrete open question is, whether the LS-2 constant can actually decrease under taking tensor powers.

Concerning LS inequalities for discrete channels there are similar open problems. Again proving analogous statements for channels with arbitrary fixed points and
determine the discrete LS constant for more channels would be interesting. New techniques that would yield bounds stable under tensor powers of discrete channels would also be of great interest.

We believe that our proofs illustrate how comparison inequality techniques can be useful and finding systematic methods to establish them, as there are for classical Markov chains, would be interesting.

\section*{Acknowledgements}
\label{sec:8}

A.\ M\"uller-Hermes and M.\ Wolf acknowledge support from the CHIST-ERA/BMBF project CQC. D. Stilck Fran\c{c}a acknowledges support from the graduate program TopMath of the Elite Network of Bavaria, the TopMath Graduate Center of TUM Graduate School at Technische Universit\"{a}t M\"{u}nchen
and the Deutscher Akademischer Austauschdienst(DAAD).
M. Wolf is also supported by the Alfried Krupp von Bohlen und Halbach-Stiftung and by grant \#48322 from the John Templeton Foundation. The opinions expressed in this publication are those of the authors and do not necessarily reflect the views of the John Templeton Foundation. 

\appendix
\section{Hypercontractivity via group theory}
\label{sec:Appendix}
Here we will consider reversible quantum dynamical semigroups with an eigenbasis consisting of unitaries commuting up to a phase. By relating such quantum semigroups to classical semigroups defined on finite abelian groups we can use the classical theory to prove hypercontractivity. We explore and extend ideas similar to \cite{junge}. In particular we get a bound on the $2\to4$ norm of tensor products of depolarizing channels. We start by reviewing some basic facts about Fourier analysis on abelian groups, the proofs of which can be found in \cite[Chapter 7]{stein2011fourier}. 

Given a finite abelian group $G$ of order $|G|$ we will denote by $V(G)$ the vector space of all functions $f:G\to\C$. For $f\in V(G)$ we define its $l_p$-norm by 
\begin{align*}
\|f\|_p^p=\frac{1}{|G|}\sum\limits_{g\in G}|f(g)|^p
\end{align*}
and for linear operators $A:V(G)\to V(G)$ we define 
\begin{align*}
\|A\|_{p\to q}=\sup\limits_{f\in V(G)}\frac{\|Af\|_q}{\|f\|_p}.
\end{align*}

The characters of $G$ are functions $\chi:G\to\{z\in\C:|z|=1\}$ such that $\chi(g_1g_2)=\chi(g_1)\chi(g_2)$ holds for any $g_1,g_2\in G$. We will denote by $\hat{G}$
the set of all characters of $G$, which is again a group isomorphic to $G$ under multiplication. The characters form an orthonormal basis for the space of functions $V(G)$ under the scalar product 
\begin{align*}
\scalar{f}{h}=\frac{1}{|G|}\sum\limits_{g\in G}f(g)^*h(g).
\end{align*}
We have $f=\sum\limits_{\chi_i\in\hat{G}}\hat{f}(i)\chi_i$ for $\hat{f}(i)=\scalar{f}{\chi_i}$. We also define $G\times G'$ to be the direct product of two groups $G,G'$ and we have 
\begin{align*}
\widehat{G\times G'}\cong\hat{G}\times\hat{G'},
\end{align*}
i.e. all characters of $G\times G'$ are of the form $\chi\chi'$ for $\chi\in\hat{G}$ and $\chi'\in\hat{G}'$

\begin{defn}[Almost commuting unitary Basis]
A set of unitaries $\{U_i\}_{0\leq i\leq d^2-1}\subset\M_d$ with $U_0=\one_d$ is called an almost commuting unitary basis (associated to an abelian group $G$) if:
\begin{enumerate}
 \item $\tr[U_i^\dagger U_j]=d\delta_{i,j}$ for all $0\leq i,j\leq d^2-1$.
 \item The $\{U_i\}_{0\leq i\leq d^2-1}$ are a projective representation of $G$, i.e. for all $0\leq i,j\leq d^2-1$ we have $U_iU_j=\phi(i,j)U_jU_i$ and $U_iU_j=\phi'(i,j)U_{i+j}$ for some $\phi'(i,j),\phi(i,j)\in\C$ with 
 $|\phi'(i,j)|=|\phi(i,j)|=1$, where in the index we mean addition in the group $G$.

\end{enumerate}
We can then associate each unitary to a character in $\hat{G}$.
 \label{defn:AlmostCommut}
\end{defn}

A prominent example of an almost commuting unitary basis is the discrete Weyl system of unitaries $\{U_{k,l}\}_{0\leq k,l\leq d}\subset\M_d$ given by:
\begin{align}
U_{k,l}=\sum\limits_{r=0}^{d-1}\nu^{rl}\ket{k+r}\bra{r},\quad\nu=e^{\frac{2i\pi}{d}}. 
\label{equ:Weyl}
\end{align}
It is easy to check the properties:
\begin{enumerate}
 \item $\tr\lb U_{i,j}^\dagger U_{k,l}\rb=d\delta_{i,k}\delta_{j,l}$ for any $i,j,k,l\in\lset 0,\ldots, d-1\rset$.
 \item $U_{i,j}U_{k,l}=\nu^{jk}U_{i+k,j+l}$
 \item $U_{k,l}^{-1}=\nu^{kl}U_{-k,-l}$ for any $k,l\in\lset 0,\ldots ,d-1\rset$.
 
\end{enumerate}
These unitaries are a projective representation of $\Z_d\times\Z_d$ in $PU(d)$. This basis has been explored in \cite{junge}
to derive similar results on hypercontractivity. However, it should be noted that there are other examples of almost commuting bases and that tensoring leads to further examples. 

By associating an almost commuting basis on $\M_d$ to the orthonormal basis of characters on $V(G)$ we can also relate norms on $\M_d$ to corresponding norms on $V(G)$. For any $X\in\M_d$ we define 
\begin{align}
f_X :=\sum\limits_{i=0}^{d^2-1}\hat{f}_X(i)\chi_i
\label{equ:fX}
\end{align}
with $\hat{f}_X(i)=\scalar{U_i}{X}_{\frac{\one}{d}}$.   

\begin{lem}\label{2normcoincide}
Let $\{U_i\}_{0\leq i\leq d^2-1}$ be an almost commuting unitary basis and $G$ the group associated to it. For any $X\in\M_d$ and $f_X$ as in \eqref{equ:fX} we have 
\begin{align*}
\|X\|_{2,\frac{\one}{d}}^2=\|f_X\|_2^2.
\end{align*}
\end{lem}
\begin{proof}
It follows immediately from the definition of an almost commuting unitary basis that it is an orthonormal basis w.r.t. $\scalar{\cdot}{\cdot}_{\frac{\one}{d}}$ and so we have:
\begin{align*}
\|X\|_{2,\frac{\one_d}{d}}^2=\scalar{X}{X}_{\frac{\one}{d}}=\sum\limits_{i=0}^{d^2-1}|\scalar{U_i}{X}_{\frac{\one}{d}}|^2=\\
\sum\limits_{i=0}^{d^2-1}|\hat{f}(i)|^2=\|f\|^2
\end{align*}
as by Plancherel's identity $\|f\|_{2}^2=\sum\limits_{i=0}^{d^2-1}|\hat{f}(i)|^2$.

\end{proof}

In the following $\liou:\M_d\to\M_d$ will denote a unital, reversible Liouvillian with spectrum $\lset \lambda_i\rset_{0\leq i\leq d^2-1}\subset\R$ and with unitary eigenvectors $\{U_i\}_{0\leq i\leq d^2-1}$  forming an almost commuting unitary basis. To such a Liouvillian we associate a classical semigroup $P_t:V(G)\to V(G)$ defined as:
\begin{align}
P_tf :=\sum\limits_{i=0}^{d^2-1}e^{\lambda_i t}\hat{f}(i)\chi_i.
\label{equ:Pt} 
\end{align}     
With this definition we can state the following theorem:

\begin{thm}\label{upperboundclassical}
Let $\liou:\M_d\to\M_d$ be a unital, reversible Liouvillian with an almost commuting unitary eigenbasis and $P_t:V(g)\to V(g)$ the associated classical semigroup as in \eqref{equ:Pt}.
Then:
\begin{align*}
\|e^{t\liou}\|_{2\to4,\frac{\one}{d}}\leq\|P_t\|_{2\to4} 
\end{align*}
\end{thm}
\begin{proof}
 Let $\{\lambda_i\}_{0\leq i\leq d^2-1}$ denote the spectrum of $\liou$ and $\{U_i\}_{0\leq i\leq d^2-1}$ the almost commuting unitary eigenbasis. Any $X\in\M_d$ can be written as $X=\sum_{i=0}^{d^2-1}\hat{f}(i)U_i$ with $\hat{f}_X(i)=\scalar{U_i}{X}_{\frac{\one}{d}}$ and we can also define $f_X$ as in \eqref{equ:fX}. Then we have:
 
\begin{align*}
\|X\|_{4,\frac{\one}{d}}^4&=\frac{1}{d}\tr[X^\dagger XX^\dagger X]=\frac{1}{d}\sum_{i_1,i_2,i_3,i_4}\hat{f}_X\lb i_1\rb^*\hat{f}_X\lb i_2\rb\hat{f}_X\lb i_3\rb^*\hat{f}_X\lb i_4\rb\tr[U^\dagger_{i_1}U_{i_2}U^\dagger_{i_3}U_{i_4}]  
\end{align*}
The unitaries $\lset U_i\rset^{d^2-1}_{i=0}$ commute and form a group up to a phase. Also by the orthogonality condition we have $\tr(U_i)=0$ for any $i\neq 0$, which implies
\begin{align*}
\frac{1}{d}|\tr[U^\dagger_{i_1}U_{i_2}U^\dagger_{i_3}U_{i_4}]|=\delta_{i_2+i_4-i_3-i_1,0}.
\end{align*}
By the triangle inequality we have:
\begin{align}\label{expression4norm2}
\|X\|_{4,\frac{\one}{d}}^4\leq\sum\limits_{i_2+i_4-i_3-i_1=0}|\hat{f}_X\lb i_1\rb\hat{f}_X\lb i_2\rb\hat{f}_X\lb i_3\rb\hat{f}_X\lb i_4\rb| 
\end{align}
Now define $f'_X\in V(G)$ as $f'_X=\sum_{i=0}^{d^2-1}|\hat{f}_X\lb i\rb|\chi_i$ and note that
\begin{align*}
\|f_X\|^4_4 = \|f'_X\|_4^4&=\frac{1}{d^2}\sum\limits_{g\in G}\left|\sum\limits_{i=0}^{d^2-1}\left|\hat{f}_X\lb i\rb\right|\chi_i\lb g\rb\right|^4\\
&=\frac{1}{d^2}\sum\limits_{g\in G}\sum_{i_1,i_2,i_3,i_4}\left|\hat{f}_X\lb i_1\rb\hat{f}_X\lb i_2\rb\hat{f}_X\lb i_3\rb\hat{f}_X\lb i_4\rb\right|\chi_{i_2+i_4-i_3-i_1}(g).
\end{align*}
where we have used the identities $\chi_i(g)^*=\chi_{-i}(g)$ and $\chi_{i_1}(g)\chi_{i_2}(g)=\chi_{i_1+i_2}(g)$ for any $g\in G$.
By the orthogonality of characters we have $\sum_{g\in G}\chi_{a}(g)=d^2\delta_{a0}$ and therefore
\begin{align*}
\|f_X\|_4^4=\sum\limits_{i_2+i_4-i_3-i_1=0}|\hat{f}_X\lb i_1\rb\hat{f}_X\lb i_2\rb\hat{f}_X\lb i_3\rb\hat{f}_X\lb i_4\rb|.  
\end{align*}
It then follows from \eqref{expression4norm2} that $\|X\|_{4,\frac{\one}{d}}^4\leq\|f_X\|_4^4$.

The semigroup $e^{t\liou}$ acts as a multiplication operator in the almost commuting unitary eigenbasis $\{U_i\}^{d^2}_{i=0}$. By construction the classical semigroup $P_t$ from \eqref{equ:Pt} associated to $e^{t\liou}$ acts in the same way in the basis of characters. This implies:
\begin{align}
\|e^{t\liou}X\|_{4,\frac{\one}{d}}\leq\|P_tf_X\|_4\leq\|P_t\|_{2\to4}\|f_X\|_2=\|P_t\|_{2\to4}\|X\|_{2,\frac{\one_d}{d}}  
\end{align}
using Lemma \ref{2normcoincide} for the last equality.

\end{proof}

Note that the proof of the previous theorem can be used for any $2\to q$ norm with $q$ an even integer~\cite{junge}.

Consider two unital, reversible Liouvillians $\liou_1:\M_{d_1}\to\M_{d_1}$ and $\liou_2:\M_{d_2}\to\M_{d_2}$ with spectra $\lset\lambda_i\rset^{d^2_1 -1}_{i=0}$ and $\lset\mu_j\rset^{d_2^2 -1}_{j=0}$ and almost commuting unitary eigenbases $\{U^1_i\}_{0\leq i\leq d_1^2-1}$ and $\{U^2_j\}_{0\leq j\leq d_2^2-1}$ associated to abelian groups $G_1$ and $G_2$. Now we can apply the above theorem to the tensor product semigroup $e^{t\Lm} = e^{t\liou_1}\otimes e^{t\liou_2}$ generated by $\Lm = \liou_1\otimes\id_{d_2}+\id_{d_1}\otimes\liou_2$. It can be verified easily that $\{U^1_i\otimes U^2_j\}_{0\leq i\leq d_1^2-1,0\leq j\leq d_2^2-1}$ is an almost commuting unitary eigenbasis for $\Lm$ associated to the abelian group $G_1\times G_2$. Let $Q_t$ denote the classical semigroup acting on $V(G_1)\otimes V(G_2)\cong V(G_1\times G_2)$ associated to $e^{t\Lm}$ as in \eqref{equ:Pt}. Also let $P^1_t$ and $P^2_t$ denote the classical semigroups associated in the same way to $e^{t\Lm_1}$ and $e^{t\Lm_2}$ respectively. Note that for any $\chi_{i,j}\in\widehat{G_1\times G_2}$ we have 
\begin{align*}
Q_t\chi_{i,j}=Q_t\chi_i\chi_j=e^{\lambda_i t}e^{\mu_j t}\chi_i\chi_j=P^1_t\otimes P^2_t\chi_{i,j} 
\end{align*}   
and hence $Q_t = P^1_t\otimes P^2_t$ as the characters (in $\widehat{G_1\times G_2}$) form a basis of $V(G_1\times G_2)$. This proves the following corollary:

\begin{cor}\label{uppertensors}
Let $\liou_1:\M_{d_1}\to\M_{d_1}$ and $\liou_2:\M_{d_2}\to\M_{d_2}$ be unital, reversible Liouvillians with almost commuting unitary eigenbases associated to abelian groups $G_1$ and $G_2$. Furthermore, let $P^1_t$ and $P^2_t$ be the associated classical semigroups as in \eqref{equ:Pt} acting on $V(G_1)$ and $V(G_2)$. Then:
\begin{align*}
\|e^{t\liou_1}\otimes e^{t\liou_2}\|_{2\to4,\frac{\one}{d_1d_2}}\leq\|P^1_t\otimes P^2_t\|_{2\to4} 
\end{align*}
\end{cor}

\bibliographystyle{alpha}
\bibliography{mybibliography1}

\end{document}